\documentclass[11pt,letterpaper]{article}
\def\withcolors{0}
\def\withnotes{0}
\usepackage{graphicx} 
\usepackage{parskip}
\usepackage{amssymb}
\usepackage{amsmath}
\usepackage{amsfonts}
\usepackage{bbm}
\usepackage{somestyle}
\usepackage{booktabs}
\usepackage[table]{xcolor}
\usepackage{pgf}
\usepackage{etoolbox}
\usepackage{palatino}
\usepackage{longtable}
\newcommand{\Vol}{\operatorname{Vol}}
\newcommand{\SRcell}[1]{%
  \ifstrequal{#1}{nan}
    {\cellcolor[HTML]{EEEEEE}\textemdash}
    {%
      \pgfmathsetmacro{\srscaled}
        {max(0,min(1,(#1-0.5)*2))}%
      \pgfmathtruncatemacro{\srchannel}
        {245*(1-\srscaled)}%
      \edef\SRapplycolor{%
        \noexpand\cellcolor[RGB]{%
          \srchannel,255,\srchannel
        }%
      }%
      \SRapplycolor #1%
    }%
}

\title{A simple and practical $o(\sqrt{n})$-time algorithm for shortest paths in power law graphs}
\date{August 2026}

\author{Jiaqi Mao\thanks{The University of Sydney. Email: \email{jmao0915@uni.sydney.edu.au}.}}

\begin{document}

\maketitle

\begin{abstract}
Computing shortest paths in large graphs is, and remains, a fundamental and practically motivated problem. While many algorithms were proposed to calculate shortest path between pairs of vertices efficiently, many of them (index-based methods) require substantial preprocessing, while others (traversal-based methods) have higher time complexity. In this paper, we propose and analyze \textsc{Pruned Bidirectional Search (PBS)}, a simple sublinear approximation algorithm for power-law graphs with parameter $\beta\in[2,3)$: our algorithm does not require any preprocessing, yet exhibits performance comparable to light index-based algorithms (of linear or sublinear index size): that is, \textsc{PBS} runs in time $O(n^{(1-1/\log\log n)/2})$ and, with high probability, returns a path with length within $\frac{41}{32}$ of the shortest path. Moreover, if one does allow a $n^{\Theta(2-1/\log\log n)}$-time preprocessing step, its query time improves to $n^{\Theta(1/\log\log n})$. We complement our theoretical results by experiments on both real-world and synthetic power-law graphs, which show that \textsc{PBS} is typically $1.84\times$--$7.76\times$ times faster than existing alternatives, while achieving an approximation ratio at most 1.05.
\end{abstract}

\section{Introduction}
Computing shortest path in large graphs is one of the most fundamental problems in graph algorithms, with an extensive and rich theory both in theory and practice. Over the past decades, many algorithms were proposed to calculate either exact shortest path or an good approximation. 

To compute the shortest path $\mathcal{P}(s,t)$ for source node $s$ and target node $t$, the simplest approach is run a breadth first search (BFS) from $s$ until $t$ is discovered. However, its expected running time is linear in number of vertices $n$, which is computationally prohibitive on large real-world networks. One popular alternative to BFS is Bidirectional BFS (BiBFS), which runs BFS on both $s$ and $t$ until their search tree intersects \cite{pohl1969bidirectional}. 
BiBFS is observed to have good expected running time on random inputs \cite{luby1989bidirectional}. 

Several variants of BiBFS were subsequently developed, to improve its performance on specific types of graphs. For instance, \cite{alon2024sublinear} shows that BiBFS finds shortest paths efficiently (namely, in time $O(\sqrt{n})$) for most pairs of vertices in expander graphs, and also combine BFS with random walks to obtain short paths for arbitrary pairs.

While this algorithm is suitable for any graph as it does not require any prior knowledge on the graph structure, it is not as fast as other algorithms which first perform some preprocessing of the input graph. An example of such algorithms is \emph{Pruned Landmark Labeling} \cite{akiba2013pruned}, which precomputes distance labels by running a sequence of pruned BFS searches, yielding a $O(1)$ running time (for each $s$-$t$ shortest path computation) after a $O(n^2)$-time preprocessing stage. Of course, as the graphs become larger, this $O(n^2)$ overhead quickly becomes computationally prohibitive; further, this approach cannot be applied to dynamic graphs.

One avenue for improvement is the fact that many real-world networks, including communication, social, and web graphs, exhibit a highly non-uniform degree distribution, a structure which could be leveraged to design faster algorithms. Empirical studies of Internet topology found that the number of vertices with a given degree approximately follows a \textit{power law} (see~\cref{{defn:powerlaw}}) \cite{faloutsos1999powerlaw}, and subsequent studies showed that almost every real world graphs' power law has parameter $\beta$ in the range $[2,3)$ \cite{Artico2020PowerLaw}. Based on this observation, Basu et al.~\cite{basu2025sublinear} recently proposed \textsc{Wormhole}, an algorithm that lies between traversal-based and index-based approaches. At a high level, \textsc{Wormhole} constructs an index only containing a sublinear number of high-degree vertices (the ``core''), to get an acceptable preprocessing time. A query first searches from $s$ and $t$ towards the core, and then connects the two searches by finding a path inside the core. Their analysis shows that one can use $o(n)$ preprocessing queries and $n^{o(1)}$ queries for each shortest-path query, while ensuring thepath returned is within an additive error of $O(\log\log n)$ of optimal (with high probability over the draw of the input graph). While this provides a practical trade off strategy between BiBFS and full distance indexes, \textsc{Wormhole} still requires a preprocessing phase, and is designed to answer multiple queries using the same stored core. 

Inspired by \textsc{Wormhole}, we propose \textsc{Pruned Bidirectional Search (PBS)} (\cref{alg:PBS}), a simpler sublinear algorithm that can be run without preprocessing, with comparable running time. It beat current $O(\sqrt{n})$ traversal-based approach, achieving $o(\sqrt{n})$. Via a more careful (yet quite clean and simple) analysis, we are able to show an $\frac{41}{32}$ multiplicative approximation guarantee rather than an additive $O(\log\log n)$. Our algorithm is not only conceptually simple, it is also easy to implement, as one only need to slightly modify BiBFS on vertex exploring~--~and achieves significantly faster running time in practice. Our main result can be stated as follows:
\begin{theorem}[Main theoretical result]\label{theo:main}
    Fix any $\beta \in (2,3)$. The algorithm \textsc{Pruned Bidirectional Search} (\cref{alg:PBS}), given query access to an $n$-vertex graph $G=(V,E)$ drawn from a power law with parameter $\beta$, as well as two vertices $s,t\in V$, runs in time $O(n^{(1-1/\log\log n)/2})$, and returns with high probability (over the choice of $G$) an $s$-$t$ path within $\frac{41}{32}$ of the shortest path.
\end{theorem}

While the above already guarantees an $o(\sqrt{n})$ time and query (with high probability) algorithm for shortest path, we observe that this can be improved much further when allowing for preprocessing of the input graph, thus enabling batch computations for various shortest paths:
\begin{theorem}[Shortest paths after preprocessing]\label{theo:preproc}
    Fix any $\beta \in [2,3)$. Given query access to an $n$-vertex graph $G=(V,E)$ drawn from a power law with parameter $\beta$ and an $O(n^{2-1/\log\log n}))$-time preprocessing phase, the algorithm \textsc{Pruned Bidirectional Search} (\cref{alg:PBS}), on input two vertices $s,t\in V$, runs in time $n^{\Theta(1/\log\log n})$, and returns with high probability (over the choice of $G$) an $s$-$t$ path within $\frac{41}{32}$ of the shortest path.
\end{theorem}
Finally, as mentioned above, \textsc{PBS} performs really well in practice, and showed amazing performance compared to \textsc{BiBFS} and \textsc{Wormhole}. We summarize the results of our experiments on real world graphs and synthetic power law graphs (detailed in \cref{sec:experiments}) in the following claim:
\begin{claim}[Performance on real graphs]
\textsc{PrunedBidirectionalSearch} is faster than BiBFS by $1.84\times$--$7.76\times$, while retaining a
two-hop success rate of at least $99.5\%$ and an average multiplicative error of at most $1.05$.
\end{claim}

\section{Preliminaries}
In what follows, an (undirected) graph $G$ is defined by an ordered pair $G = (V, E)$: unless specified otherwise, we write $n=|V|$ for the number of vertices. We let $\delta(v)=\{u\in V:(u,v)\in E\}$.

\begin{definition}[Power Law]\label{defn:powerlaw}
Given a graph $ G = (V,E)$ with $n$ vertices and $m$ edges, with constant $\beta$, it is said to follow power law if

$$ |\{\deg(v) = k\ |\ v \in V \}| \varpropto 1/k^\beta $$
\end{definition}

To analyze approximation algorithms performance on power-law graphs, we follow the random graph generation model proposed by Chung and Lu \cite{pnas/2002/CL}. Under this model, each vertex $v_i$ is assigned a weight $w_i$. Without loss of generality, we can write all weights as a non-decreasing weight sequence $\mathcal{W} = \{w_1,w_2, ..., w_n\}$. For each pair of vertices $(v_i,v_j)$, an edge is formed between them independently with probability
\[
    p_{ij}=\frac{w_iw_j}{W},
    \qquad
    W=\sum_{i=1}^n w_i,
\]
where $\max_i w_i^2\leq W$. Under such setting, \textit{Volume} of a vertex set $S\subseteq V$ is defined by $\Vol(S) = \sum_{u \in S}, w_u$, and define second order average degree as
\[
    \Vol_k(S)=\sum_{v_i\in S}w_i^k,
    \qquad
    \Vol(S)=\Vol_1(S),
    \qquad
    \widetilde d=\frac{\Vol_2(V)}{\Vol(V)}.
\]

\begin{definition}[Power-law sequence]\label{defn:powerlaw}
An expected degree sequence follows a power law with exponent $\beta$
if
\[
    |\{v_i:w_i\geq x\}|=\Theta(nx^{1-\beta})
\]
\end{definition}

Then we will introduce some important results from Chung and Lu \cite{Lu2002ProbabilisticMI, pnas/2002/CL}, which will be used in our theoretical analysis.

\begin{lemma}[Edges between vertex sets, Lemma 9 in \cite{Lu2002ProbabilisticMI}]\label{lem:set-edge}
Let $S,T\subseteq V$ be disjoint. For every $c>0$, if
\[
    \frac{\Vol(S)\Vol(T)}{W}\geq c,
\]
then
\[
    \Pr\!\left(\operatorname{dist}(S,T)>1\right)\leq e^{-c}.
\]
\end{lemma}

For $2<\beta<3$, let
\[
    \gamma=\frac{1}{\log\log n},
    \qquad
    t=n^\gamma,
    \qquad
    S_C=\{v_i:w_i\geq t\}.
\]
We call $S_C$ the core of $G$.

\begin{lemma}[Core diameter, Theorem 4 in \cite{pnas/2002/CL}] \label{lem:core}
For every fixed $2<\beta<3$, the core satisfies
\[
    |S_C|=\Theta(nt^{1-\beta})
    =n^{1-\gamma(\beta-1)+o(1)}
\]
and, with high probability, $S_C$ is completely connected and 
\[
    \operatorname{diam}(G[S_C])=O(\log\log n).
\]
\end{lemma}

\paragraph{Table of Notations.}

\begin{longtable}{@{}p{0.24\textwidth}p{0.71\textwidth}@{}}
\label{tab:notation}\\
\toprule
Notation & Meaning \\
\midrule
\endfirsthead

\toprule
Notation & Meaning \\
\midrule
\endhead

\bottomrule
\endfoot

$G=(V,E)$
& An undirected graph with vertex set $V$ and edge set $E$. \\

$\deg(v)$
& Degree of vertex $v$. \\

$\delta(v)$
& Set of neighbours of vertex $v$. \\

$\mathcal{P}(u,v)$
& A path from $u$ to $v$. \\

$\dist(u,v)$
& Shortest-path distance between vertices $u$ and $v$. \\

$\dist(A,v)$
& Distance from a vertex set $A$ to $v$, defined as
$\min_{u\in A}\dist(u,v)$. \\

$\beta$
& Power-law exponent. This paper considers $\beta\in[2,3)$. \\

$N_k$
& Number of vertices having degree $k$:
$N_k=|\{v\in V:\deg(v)=k\}|=\Theta(nk^{-\beta})$. \\

$\Delta$
& Maximum degree in $G$, namely $\Delta=\max_{v\in V}\deg(v)$. \\

$\mathcal{W}=\{w_1,\ldots,w_n\}$
& Non-decreasing weight sequence in the Chung--Lu model. \\

$w_i$
& Weight, or expected degree, assigned to vertex $v_i$. \\

$W$
& Total vertex weight, defined as $W=\sum_{i=1}^{n}w_i$. \\

$\rho$
& Chung--Lu normalization factor, $\rho=1/W$, so that
$\Pr[(v_i,v_j)\in E]=w_iw_j\rho$. \\

$d$
& Average degree of $G$, defined as $d=2m/n$. \\

$\widetilde d$
& Second order average degree of $G$, defined as $\widetilde d =\frac{\Vol_2(V)}{\Vol(V)}$ \\

$\gamma$
& Core threshold exponent, defined as
$\gamma=1/\log\log n$. \\

$S_C$
& High-degree core:
$S_C=\{v\in V:\deg(v)\ge n^\gamma\}$. \\

$\Vol(S)$
& Degree volume of a vertex set $S$:
$\Vol(S)=\sum_{v\in S}\deg(v)$. \\

$S_{\mathrm{in}}(v)$
& Neighbours of $v$ that belong to the core:
$\{u\in\delta(v):\deg(u)\ge n^\gamma\}$. \\

$S_{\mathrm{out}}(v)$
& Neighbours of $v$ outside the core:
$\{u\in\delta(v):\deg(u)<n^\gamma\}$. \\

$D_{\delta(u)}$
& Total degree volume of the neighbours of $u$:
$D_{\delta(u)}=\sum_{v\in\delta(u)}\deg(v)$. \\

$\Bbb{T}_s,\Bbb{T}_t$
& Search trees rooted at $s$ and $t$, respectively. \\

$C(s),C(t)$
& First core vertices selected by the searches from $s$ and $t$. \\

$\mathrm{HQ}$
& High Priority Queue containing neighbours selected by the
$\frac34$-volume pruning rule. \\

$\mathrm{LQ}$
& Low Priority Queue containing the remaining neighbours. \\

$T_i$
& Depth-$i$ BFS frontier grown from the target vertex during the
approximation analysis. \\

$\epsilon$
& Expected approximation error contributed by a pruning step. \\

\end{longtable}

\section{Algorithm}

Our algorithm, \textsc{Pruned Bidirectional Search (PBS)} (which is given as pseudocode in~\cref{alg:PBS}) consists of 2 phases:
\begin{description}
    \item[Pruning Phase:]
    Similar to BiBFS, PBS first starts an bidirectional search on both $s$ and $t$. When exploring a vertex $u$, PBS sort $\delta(u)$ by degree and calculate the sum of their degree $D_{\delta(u)} = \sum_{v \in \delta(u)} \deg(v)$. Then, PBS takes neighbour vertices from high degree to low to a High Priority Queue (HQ) until the degree sum of taken vertices reaches $\frac{3}{4}D_{\delta(u)}$, and put other vertices in $\delta(u)$ to a Low Priority Queue (LQ). As long as their is any vertices in HQ, PBS only expand the search tree with vertices in HQ regardless of tree height. Note that we still pick the smallest height vertex when choosing next vertex inside HQ, which is similar to BFS.\\
    We note search tree of $s$ and $t$ by $\Bbb{T}_s$ and $\Bbb{T}_t$ respectively. PBS stops expanding $\Bbb{T}_i$ when (a) $\Bbb{T}_s$ intersects with $\Bbb{T}_t$, in which case it terminates and return; or (b) $\Bbb{T}_i$ reached a vertex $C(i)$ with degree not less than $n^\gamma$, where $\gamma = 1/\log\log n$. When PBS stops expanding both $\Bbb{T}_s$ and $\Bbb{T}_t$, it transit to Core-Routing Phase with $C(s)$ and $C(t)$.
    \item[Core-Routing Phase:]
    Once $\Bbb{T}_s$ and $\Bbb{T}_t$ both reached the core, PBS starts routing through the core. By a BiBFS from $C(s)$ and $C(t)$ but only search vertices with degree not less than $n^\gamma$. Since the core is connected, the shortest path inside the core will be eventually found. Then, PBS concatenates 3 segments, $s-C(s)-C(t)-t$, and return it as an approximated path.\\
    In extremely rare situations (never seen in experiments), it fails to find a path inside the core. Then it continues to expand $\Bbb{T}_s$ and $\Bbb{T}_t$ until all vertices in same connected component are visited, which is same as the worst case of BiBFS.
\end{description}

\begin{algorithm}[htbp]
\caption{Pruned Bidirectional Search}
\label{alg:PBS}
\begin{algorithmic}[1]
\Require A Graph $G$ source node $s$ and target node $t$ 
\State $\Bbb{T}(s) \gets \{s\}$, $\Bbb{T}(t) \gets \{t\}$
\State $C(s) \gets \emptyset$, $C(t) \gets \emptyset$
\While{$C(s) = \emptyset$ or $C(t) = \emptyset$ }
    \For {$i \in {s,t}$}
        \State Expand $\Bbb{T}(i)$ by one level
        \For {$u \in$ new vertices}
            \If{$u$ is in top 3/4 volume}
            \State HQ $\gets$ HQ + $u$
            \Else
            \State LQ $\gets$ LQ + $u$
            \EndIf
        \EndFor
        \State $C(i) \gets \Bbb{T}(i)\cap S_C$
    \EndFor
    \If{$\Bbb{T}(s)\cap\Bbb{T}(t) \neq \emptyset$}
        \State return FindPath($\Bbb{T}(s)$,$\Bbb{T}(t)$)
    \EndIf
\EndWhile

\State \Return $\mathcal{P}(s,C(s))$ + BiBFS($G$, $C(s)$, $C(t)$) + $\mathcal{P}(C(t),t)$ 
\end{algorithmic}
\end{algorithm}

\section{Theoretical Analysis}
In this section, we prove~\cref{theo:main}, by first establishing in~\cref{ssec:time} the claimed time complexity, before analyzing in~\cref{ssec:approx} the quality of the path returned by the algorithm. We then conclude the section by establishing~\cref{theo:preproc} in~\cref{ssec:preproc}.
\subsection{Time Complexity}\label{ssec:time}
\begin{enumerate}
    \item Pruning Phase \\
    If the algorithm terminates early, it takes less time than reaching the Core on both side. Hence, we can upper bound time complexity of this phase by time of reaching the Core. To analyze this, consider we are expanding the BFS tree $\Bbb{T}(s)$.  When exploring a new edge, the probability it reaches the Core is at least 
    $$
    p = \sum_{u \in S_C} \frac{\Vol(S_C)}{2m} \geq \frac{|S_c|\cdot n^\gamma} { n\cdot d}
    $$
    Since $\beta > 2$, we have the average degree of graph $$d = \sum_{k=1}^\infty \frac{1}{k^\beta} < \sum_{k=1}^\infty \frac{1}{k^2} = \frac{\pi^2}{6}$$
    Then we have
    $$
    p \geq \frac{n \cdot \frac{1}{n^{\gamma(\beta-1)}}}{n\cdot \frac{\pi^2}{6}} = \Theta\mleft(\frac{1}{n^{\gamma(\beta-1)}}\mright)
    $$
    After $O(1/p)$ steps, we reach the core with constant probability. Thus, expected time complexity of this phase is $O(n^{\gamma(\beta-1)}) = O(n^{\Theta(1/\log \log n)})$.

    \item Core-Routing Phase\\
    The size of the Core satisfies (the final inequality follows from $\beta\geq2$)
    $$
    \begin{aligned}
    |S_C|
    &= O\left(n\sum_{i=k}^{\infty}\frac{1}{i^\beta}\right)\\
    &\leq O\left(n\int_{k-1}^{\infty}x^{-\beta}\,dx\right) =O\left(\frac{n(k-1)^{1-\beta}}{\beta-1}\right)\\
    &=O\left(n^{1-\gamma(\beta-1)}\right)
     \leq O\left(n^{1-\gamma}\right),
    \end{aligned}
    $$
    Under such condition, BiBFS inside the core runs in $O(\sqrt{|S_C|} = O(n^{(1-\gamma)/2})$ expected time \cite{alon2024sublinear}.
\end{enumerate}

Therefore, the overall time complexity is $O(n^{(1-\gamma)/2})$.

\subsection{Approximation Quality}\label{ssec:approx}
To simplify our analysis, we want to show that 

\begin{claim} \label{claim:coreroute}
 When $2 \leq \beta < 3$ and doing BFS inside the Core on vertex $u$, 
 $$
 \Vol(\{v: v \in \delta(u), v\in S_C\}) \geq \frac{3}{4} \Vol(\delta(u))
 $$ 
\end{claim}

\begin{proof}
    For any vertex in the core $v \in S_C$, we have $\deg(v) \geq n^\gamma$. We define $S_{in}(v) = \{u: u\in\delta(v), \deg(u)\geq n^\gamma\}$ as neighbours of $u$ in the core, and $S_{out}(v) = \{u: u\in\delta(v), \deg(u)< n^\gamma\}$ respectively. \\
    Under the graph generation model, An edge between $(u,v)$ is formed with probability $\frac{\deg(u)\deg(v)}{W}$.
    Then the expected volume of $S_{in}(v)$ and $S_{out}(v)$ are 
    $$
    \begin{aligned}
        \mathbb E[\Vol(S_{\rm in}(v))]
        &=
        \frac{\deg(v)}{W}
        \sum_{\substack{u\in V\\ \deg(u)\ge t}}\deg(u)^2,\\
        \mathbb E[\Vol(S_{\rm out}(v))]
        &=
        \frac{\deg(v)}{W}
        \sum_{\substack{u\in V\\ \deg(u)<t}}\deg(u)^2.
    \end{aligned}
    $$
    Cancelling same factors gives us
    $$
        \frac{\mathbb E[\Vol(S_{\rm in}(v))]}{\mathbb E[\Vol(S_{\rm out}(v))]}
        =
        \frac{\sum_{\deg(u)\ge n^\gamma}\deg(u)^2}{\sum_{\deg(u)<n^\gamma}\deg(u)^2}
    $$
    By taking in the number of vertices with degree $k$ is $N_k = \Theta(nk^{-\beta})$ and letting $\Delta = n^{\Theta(1)}$ to be highest degree, we get
    $$
        \frac{\mathbb E[\Vol(S_{\rm in}(v))]}{\mathbb E[\Vol(S_{\rm out}(v))]}
        = \frac{\Theta(n\cdot(\Delta^{3-\beta}-n^{\gamma(3-\beta)}))}{\Theta(n\cdot n^{\gamma(3-\beta)})}
        = \Theta((\frac{\Delta}{n^\gamma})^{3-\beta} ) = n^{\Theta(1-\gamma)}
    $$
    Then, with high probability, the volume ratio $\frac{\Vol(S_{\rm in}(v))}{\Vol(\delta(v))} \geq \frac{3}{4}$
\end{proof}

By \cref{claim:coreroute}, the approximation quality of PBS is not worse than purely using $\frac{3}{4}$ pruned search from $s,t$ until they intersect as Core-Routing phase has a bigger search space. So, we can simplify our analysis by treating both phases as $\frac{3}{4}$ pruning search.

Consider any step expanding $\Bbb{T}(s)$ at vertex $u$, note the set of vertices will be pushed into HQ as $A$, and the set of vertices will be pushed into LQ as $B$. Approximation error can only introduced when $A$ and $B$ both have a path to $t$, and $\mathcal{P}(B,t)$ is shorter than $\mathcal{P}(A,t)$.\\
To compare $\dist(A,t)$ and $\dist(B,t)$, imagine we grow a classic BST from $t$, note as $T_i$ at $i$-th depth. Conditioned on $T_i$, the probability of $T_{i+1}$ reaches $A$ and $B$ is 
$$
    \mathrm{Pr}[T_{i+1} \cap A \neq \emptyset] = 1 - \prod_{a \in A} (1-w_a\Vol(T)\rho) \doteq \Vol(A)\Vol(T)\rho
$$

$$
    \mathrm{Pr}[T_{i+1} \cap B \neq \emptyset] = 1 - \prod_{b \in B} (1-w_b\Vol(T)\rho) \doteq \Vol(B)\Vol(T)\rho
$$
If an error was introduced, its volume is $\dist(A,T_{i+1})$ in $G \backslash u$. By \cref{lem:set-edge}, 
$$
    \mathrm{Pr}[\dist(A,T_{i+1}) > 1] < e^{-\Vol(A)\Vol(T_{i+1})\rho} < e^{-\Vol(A)\Vol(T_{i})d\rho}
$$
Conditioned on reaching $A$ or $B$ at $T_{i+1}$, we have 
$$
\mathrm{Pr}[\dist(A,T_{i+1} ) \leq 1 | T_{i+1} \cap (A\cup B) \neq \emptyset] > \frac{1-e^{-\Vol(A)\Vol(T_{i+1})\rho}}{\Vol(A\cup B)\Vol(T)\rho} \geq \frac{\Vol(A)\Vol(T_{i})d\rho}{\Vol(A\cup B)\Vol(T_{i})\rho} \doteq \frac{3d}{4}
$$
Also, we can bound the probability of $\dist(A,T_{i+1} ) > k$ by further expanding $T$. Since $\beta < 3$, we have average degree $d = \zeta(\beta) > \zeta(3) > 1.2$. Applying this, we have the expectation of $\dist(A,T_{i+1})$ when $T_{i+1}$ intersect with $B$
$$
    \begin{aligned}
    \Bbb{E}[\dist(A,T_{i+1})] &\leq \sum_{k=1} k\cdot(\mathrm{Pr}[\dist(A,T_{i+1} )  \\
    &\leq k] - \mathrm{Pr}[\dist(A,T_{i+1} ) \leq k-1]) \\
    &\leq \frac{3}{4}+ \sum_{k=1}\frac{1}{4\cdot 1.2^k}  \leq \frac{9}{4}
    \end{aligned}
$$
The probability of $T_{i+1}$ intersect only with $B$ is approximately
$$
    \begin{aligned}
    \mathrm{Pr}[T_{i+1} \cap B \neq \emptyset,T_{i+1} \cap A = \emptyset] &\doteq \frac{\Vol(B)}{\Vol(A\cup B)} 
    \cdot (1 - \frac{\Vol(A)}{\Vol(A\cup B)})\cdot \mathrm{Pr}[T_{i+1} \cap (A\cup B) \neq \emptyset] \\
    &\geq \frac{1}{16}\mathrm{Pr}[T_{i+1} \cap (A\cup B) \neq \emptyset]
    \end{aligned}
$$
And we can recursively perform this analysis when $T_{i+1}$ reaches neither. Then we have the expected error introduced on each node along the path is
$$
    \epsilon = \sum \frac{1}{16}\mathrm{Pr}[T_{i+1} \cap (A\cup B) \neq \emptyset]\cdot \Bbb{E}[\dist(A,T_{i+1})] \leq \frac{9}{64}
$$
Here we can conclude PBS returns a $(1+\frac{9}{32})$ approximation path with high probability.

\subsection{Proof of~\cref{theo:preproc}}\label{ssec:preproc}
To conclude this section, we establish~\cref{theo:preproc}, by analyzing how a preprocessing stage can improve the query time.

By \cref{lem:core}, there is almost surely a unique connected core in $G$, with size of $O(n^{1-\gamma})$. Getting all vertices in the core require $O(n)$ queries. Pruned Landmark Labeling \cite{akiba2013pruned} could run on the core with $O(n^{2-\gamma})$ time complexity. Having index built on the core, Core-Routing phase can be done in $O(1)$ time. Therefore, the overall time complexity for one path query is $n^{\Theta(1/\log \log n)}$ and the approximation guarantee holds.

\section{Experiments}
\label{sec:experiments}

\paragraph{Setup.}
We evaluate \textsc{PrunedBidirectionalSearch} (\textsc{PBS}) on five
real-world networks from SNAP and three synthetic power-law graphs. All source code is publicly available at \url{https://github.com/JiaqiM26/PrunedBidirectionalSearch}. We used the implementation of \textsc{BiBFS} and \textsc{Wormhole} from Basu et al. \cite{basu2025sublinear}.
Each synthetic
graph contains one million vertices, with power-law exponent
$\beta \in \{2.0,2.5,2.9\}$. We execute 1,000 random source--target
queries per graph. Exact bidirectional BFS (\textsc{BiBFS}) provides the
ground-truth distances. We report total query time and the fraction of
answers within $k$ hops of the exact distance, denoted by $+k$ Success Rate.

\begin{figure}[t]
    \centering
    \includegraphics[width=\linewidth]
        {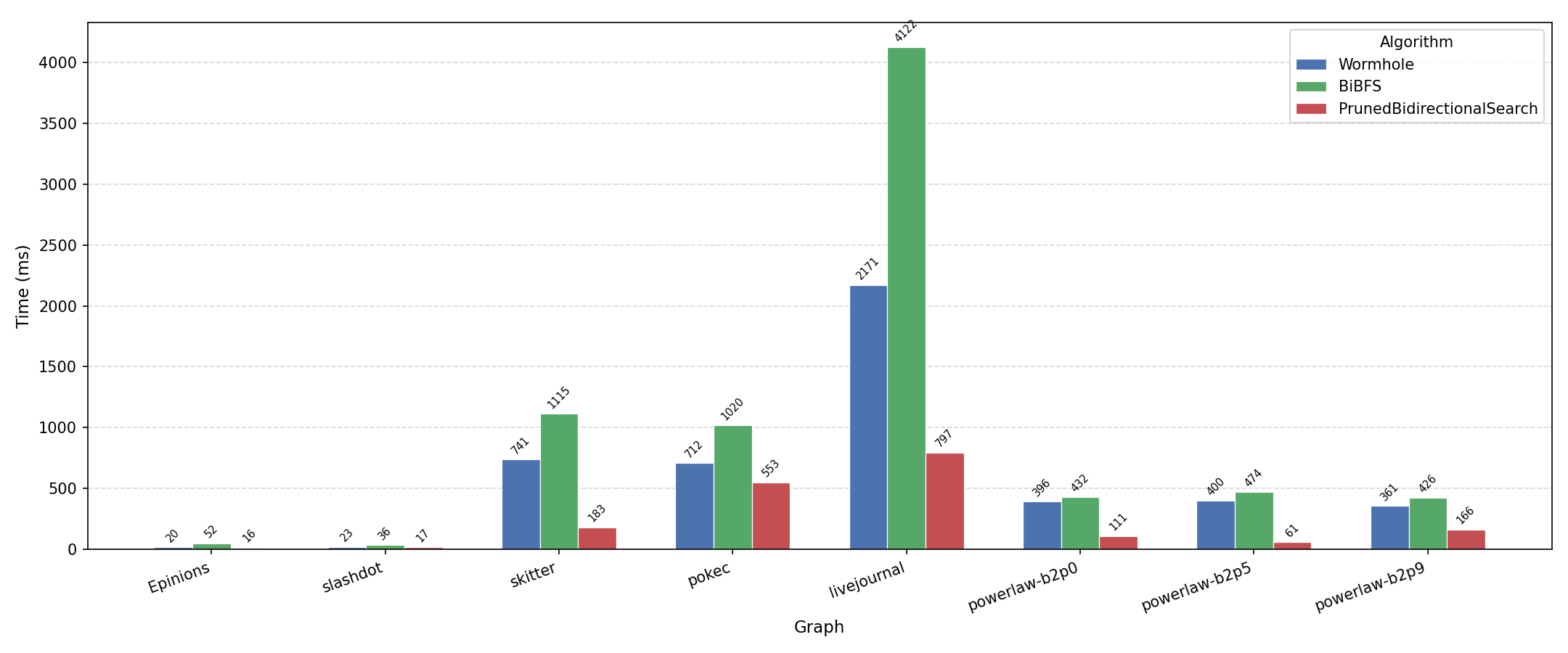}
    \caption{Average running time for 1,000 queries on each graph.}
    \label{fig:running-time}
\end{figure}

\paragraph{Results.}
The experimental results in \cref{fig:running-time,tab:query-results}
show that \textsc{PBS} is the fastest method on every graph. It achieves
a $1.84\times$--$7.76\times$ speedup over \textsc{BiBFS} and a
$1.28\times$--$6.55\times$ speedup over \textsc{Wormhole}. Its average
additive error ranges from $0.031$ to $0.321$ hops, while its average
multiplicative error ranges from $1.008\times$ to $1.050\times$. At
least $96.2\%$ of its estimates are within one hop of the exact distance,
and at least $99.5\%$ are within two hops.

\begin{table*}[t]
\centering
\small
\setlength{\tabcolsep}{2.7pt}
\begin{tabular}{lrrrrrrrrrrr}
\toprule
& \multicolumn{1}{c}{\textsc{BiBFS}}
& \multicolumn{5}{c}{\textsc{Wormhole}}
& \multicolumn{5}{c}{\textsc{PrunedBidirectionalSearch}}\\
\cmidrule(lr){2-2}\cmidrule(lr){3-7}\cmidrule(lr){8-12}
Graph
& Time
& Time & $+0$ & $+1$ & $+2$ & AME
& Time & $+0$ & $+1$ & $+2$ & AME\\
\midrule
Epinions
& 52.4
& 20.2 & \SRcell{.980} & \SRcell{1.000} & \SRcell{1.000} & 1.006
& \textbf{15.8} & \SRcell{.968} & \SRcell{1.000} & \SRcell{1.000} & 1.008\\
Slashdot
& 35.9
& 22.5 & \SRcell{.967} & \SRcell{.999} & \SRcell{1.000} & 1.010
& \textbf{17.4} & \SRcell{.942} & \SRcell{1.000} & \SRcell{1.000} & 1.016\\
Skitter
& 1115.1
& 740.7 & \SRcell{.942} & \SRcell{.999} & \SRcell{1.000} & 1.013
& \textbf{183.1} & \SRcell{.882} & \SRcell{.998} & \SRcell{1.000} & 1.025\\
Pokec
& 1020.0
& 711.9 & \SRcell{.516} & \SRcell{.923} & \SRcell{.998} & 1.128
& \textbf{553.0} & \SRcell{.867} & \SRcell{.996} & \SRcell{1.000} & 1.032\\
LiveJournal
& 4122.2
& 2171.5 & \SRcell{.714} & \SRcell{.982} & \SRcell{.998} & 1.058
& \textbf{796.6} & \SRcell{.832} & \SRcell{.995} & \SRcell{1.000} & 1.034\\
Power-law $\beta=2.0$
& 432.0
& 395.5 & \SRcell{.988} & \SRcell{1.000} & \SRcell{1.000} & 1.003
& \textbf{111.2} & \SRcell{.969} & \SRcell{1.000} & \SRcell{1.000} & 1.009\\
Power-law $\beta=2.5$
& 473.7
& 400.1 & \SRcell{.932} & \SRcell{.998} & \SRcell{1.000} & 1.015
& \textbf{61.1} & \SRcell{.859} & \SRcell{.998} & \SRcell{1.000} & 1.030\\
Power-law $\beta=2.9$
& 426.3
& 360.7 & \SRcell{.683} & \SRcell{.941} & \SRcell{.997} & 1.056
& \textbf{165.9} & \SRcell{.723} & \SRcell{.962} & \SRcell{.995} & 1.050\\
\bottomrule
\end{tabular}
\caption{Total running time in milliseconds for 1,000 queries, success
rates within 0, 1, and 2 hops of the exact distance, and average
multiplicative error (AME). Bold values indicate the lowest runtime for
each graph.}
\label{tab:query-results}
\end{table*}

\printbibliography

@article{
pnas/2002/CL,
author = {Fan Chung  and Linyuan Lu },
title = {The average distances in random graphs with given expected degrees},
journal = {Proceedings of the National Academy of Sciences},
volume = {99},
number = {25},
pages = {15879-15882},
year = {2002},
doi = {10.1073/pnas.252631999},
URL = {https://www.pnas.org/doi/abs/10.1073/pnas.252631999},
eprint = {https://www.pnas.org/doi/pdf/10.1073/pnas.252631999}}

@phdthesis{pohl1969bidirectional,
  author       = {Ira Pohl},
  title        = {Bi-directional and heuristic search in path problems},
  school       = {Stanford University, {USA}},
  year         = {1969},
  url          = {https://searchworks.stanford.edu/view/2197829},
  bibsource    = {dblp computer science bibliography, https://dblp.org}
}

@article{luby1989bidirectional,
  author       = {Michael Luby and
                  Prabhakar Ragde},
  title        = {A Bidirectional Shortest-Path Algorithm with Good Average-Case Behavior},
  journal      = {Algorithmica},
  volume       = {4},
  number       = {4},
  pages        = {551--567},
  year         = {1989},
  url          = {https://doi.org/10.1007/BF01553908},
  doi          = {10.1007/BF01553908},
  bibsource    = {dblp computer science bibliography, https://dblp.org}
}

@inproceedings{alon2024sublinear,
  author       = {Noga Alon and
                  Allan Gr{\o}nlund and
                  S{\o}ren Fuglede J{\o}rgensen and
                  Kasper Green Larsen},
  editor       = {Rastislav Kr{\'{a}}lovic and
                  Anton{\'{\i}}n Kucera},
  title        = {Sublinear Time Shortest Path in Expander Graphs},
  booktitle    = {49th International Symposium on Mathematical Foundations of Computer
                  Science, {MFCS} 2024, Bratislava, Slovakia, August 26-30, 2024},
  series       = {LIPIcs},
  volume       = {306},
  pages        = {8:1--8:13},
  publisher    = {Schloss Dagstuhl - Leibniz-Zentrum f{\"{u}}r Informatik},
  year         = {2024},
  url          = {https://doi.org/10.4230/LIPIcs.MFCS.2024.8},
  doi          = {10.4230/LIPICS.MFCS.2024.8},
  bibsource    = {dblp computer science bibliography, https://dblp.org}
}

@inproceedings{akiba2013pruned,
  author       = {Takuya Akiba and
                  Yoichi Iwata and
                  Yuichi Yoshida},
  editor       = {Kenneth A. Ross and
                  Divesh Srivastava and
                  Dimitris Papadias},
  title        = {Fast exact shortest-path distance queries on large networks by pruned
                  landmark labeling},
  booktitle    = {Proceedings of the {ACM} {SIGMOD} International Conference on Management
                  of Data, {SIGMOD} 2013, New York, NY, USA, June 22-27, 2013},
  pages        = {349--360},
  publisher    = {{ACM}},
  year         = {2013},
  url          = {https://doi.org/10.1145/2463676.2465315},
  doi          = {10.1145/2463676.2465315},
  bibsource    = {dblp computer science bibliography, https://dblp.org}
}

@inproceedings{faloutsos1999powerlaw,
  author       = {Michalis Faloutsos and
                  Petros Faloutsos and
                  Christos Faloutsos},
  editor       = {Lyman Chapin and
                  James P. G. Sterbenz and
                  Guru M. Parulkar and
                  Jonathan S. Turner},
  title        = {On Power-law Relationships of the Internet Topology},
  booktitle    = {Proceedings of the {ACM} {SIGCOMM} 1999 Conference on Applications,
                  Technologies, Architectures, and Protocols for Computer Communication,
                  August 30 - September 3, 1999, Cambridge, Massachusetts, {USA}},
  pages        = {251--262},
  publisher    = {{ACM}},
  year         = {1999},
  url          = {https://doi.org/10.1145/316188.316229},
  doi          = {10.1145/316188.316229},
  bibsource    = {dblp computer science bibliography, https://dblp.org}
}

@inproceedings{basu2025sublinear,
  author       = {Sabyasachi Basu and
                  Nadia Koshima and
                  Talya Eden and
                  Omri Ben{-}Eliezer and
                  C. Seshadhri},
  editor       = {Wolfgang Nejdl and
                  S{\"{o}}ren Auer and
                  Meeyoung Cha and
                  Marie{-}Francine Moens and
                  Marc Najork},
  title        = {A Sublinear Algorithm for Approximate Shortest Paths in Large Networks},
  booktitle    = {Proceedings of the Eighteenth {ACM} International Conference on Web
                  Search and Data Mining, {WSDM} 2025, Hannover, Germany, March 10-14,
                  2025},
  pages        = {20--29},
  publisher    = {{ACM}},
  year         = {2025},
  url          = {https://doi.org/10.1145/3701551.3703512},
  doi          = {10.1145/3701551.3703512},
  bibsource    = {dblp computer science bibliography, https://dblp.org}
}

@phdthesis{ Lu2002ProbabilisticMI,
  title={Probabilistic methods in massive graphs and internet computing},
  author={Linyuan Lu and Fan Chung Graham},
  year={2002},
  school={University of California, San Diego},
  url={https://search-library.ucsd.edu/permalink/01UCS_SDI/1be9fsd/alma991021291089706535}
}

@article{Artico2020PowerLaw,
    author = {Artico, I. and Smolyarenko, I. and Vinciotti, V. and Wit, E. C.},
    title = {How rare are power-law networks really?},
    journal = {Proceedings of the Royal Society A: Mathematical, Physical and Engineering Sciences},
    volume = {476},
    number = {2241},
    pages = {20190742},
    year = {2020},
    month = {09},
    issn = {1364-5021},
    doi = {10.1098/rspa.2019.0742},
    url = {https://doi.org/10.1098/rspa.2019.0742},
    eprint = {https://royalsocietypublishing.org/rspa/article-pdf/doi/10.1098/rspa.2019.0742/637460/rspa.2019.0742.pdf},
}

\end{document}